\documentclass[lineno]{biometrika}

\usepackage{natbib,setspace}
\usepackage{times}

\usepackage[plain,noend]{algorithm2e}

\makeatletter
\renewcommand{\algocf@captiontext}[2]{#1\algocf@typo. \AlCapFnt{}#2} 
\def\@algocf@capt@plain{top}
\renewcommand{\algocf@makecaption}[2]{%
  \addtolength{\hsize}{\algomargin}%
  \sbox\@tempboxa{\algocf@captiontext{#1}{#2}}%
  \ifdim\wd\@tempboxa >\hsize
    \hskip .5\algomargin%
    \parbox[t]{\hsize}{\algocf@captiontext{#1}{#2}}
  \else%
    \global\@minipagefalse%
    \hbox to\hsize{\box\@tempboxa}
  \fi%
  \addtolength{\hsize}{-\algomargin}%
}
\makeatother

\def\E{\mathbb{E}}
\def\E{\mathbbmss{E}}

\def\I{\mathbbmss{1}}

\def\tildealy{\widetilde \alpha_y}
\def\tildeald{\widetilde \alpha_d}

\def\CiD{\stackrel{d}{\longrightarrow}}
\def\CiP{\stackrel{p}{\longrightarrow}}

\def\bx{\mathbf{x}}

\def\bU{\mathbf{U}}
\def\bu{\mathbf{u}}

\def\bX{\mathbf{X}}

\usepackage{amscd,amsmath,amsfonts,verbatim,amscd}
\usepackage{mathrsfs,graphics,graphicx}
\usepackage{bbm}
\usepackage[all]{xy}
\usepackage{xcolor}
\usepackage{bm}
\usepackage{bbm}
\usepackage{amssymb,caption,fancyhdr}

\begin{document}

\nolinenumbers
\markboth{A. Ertefaie}{Variable Selection in Causal Inference}

\title{Variable Selection in Causal Inference using a Simultaneous Penalization Method}

\author{Ashkan Ertefaie$^{1}$, Masoud Asgharian$^{2}$ and David A. Stephens$^{2}$}
\affil{$^{1}$Department of Statistics, University of Pennsylvania, Philadelphia, PA, USA\\
$^{2}$Department of Mathematics and Statistics, McGill University, Montreal, Canada
 \email{ertefaie@wharton.upenn.edu} }

\maketitle

\begin{abstract}
In the causal adjustment setting, variable selection techniques based on one of either the outcome or treatment allocation model can result in the omission of confounders, which leads to bias, or the inclusion of spurious variables, which leads to variance inflation, in the propensity score.  We propose a variable selection method based on a penalized objective function which considers the outcome and treatment assignment models simultaneously. The proposed method facilitates confounder selection in high-dimensional settings. We show that under regularity conditions our method attains the oracle property. The selected variables are used to form a doubly robust regression estimator of the treatment effect.   We show that under some conditions our method attains
the oracle property.
Simulation results are presented and economic growth data are analyzed. Specifically, we study the effect of life expectancy as a measure of population health on the average growth rate of gross domestic product per capita.
\end{abstract}


\section{Introduction} \label{intro}

In the analysis of observational data, when attempting to establish the magnitude of the causal effect of treatment (or exposure) in the presence of confounding, the practitioner is faced with certain modeling decisions that facilitate estimation.  Should one take the parametric approach, at least one of two statistical models must be proposed; (i) the \textit{conditional mean model} that models the expected outcome as a function of predictors, and (ii) the \textit{treatment allocation model} that describes the mechanism via which treatment is allocated to (or, at least, received by) individuals in the study, again as a function of the predictors \citep{rosenbaum1983central, robins2000marginal}.

Predictors that appear in both mechanisms (i) and (ii) are termed \textit{confounders}, and their omission from model (ii) is typically regarded as a serious error, as it leads to inconsistent estimators of the treatment effect. Thus practitioners usually adopt a conservative approach, and attempt to ensure that they do not omit confounders by fitting a richly parameterized treatment allocation model. The conservative approach, however, can lead to predictors of treatment allocation only -- and not outcome -- being included in the treatment allocation model. The inclusion of such ``spurious" \emph{instrumental} variables in model (ii) is usually regarded as harmless.  However, the typical forfeit for this conservatism is inflation of variance of the effect estimator \citep{greenland2008invited, schisterman2009overadjustment}.  This problem also applies to the conditional mean model, but is in practice less problematic, as practitioners seem to be more concerned with bias removal, and therefore more likely to introduce the spurious variables in model (ii).  Little formal guidance as to how the practitioner should act in this setting has been provided.

As has been conjectured and studied in simulation by \cite{brookhart2006variable}, it is plausible that judicious variable selection may lead to appreciable efficiency gains, and several approaches with this aim have been proposed.  However, confounder selection methods based on either just the treatment assignment model or just the outcome model may fail to account for non-ignorable confounders which barely predict the treatment or the outcome, respectively \citep{crainiceanu2008adjustment, belloni2014inference}: in this manuscript, we use the term \textit{weak confounder} for these variables. \cite{vansteelandt2010model} shows that confounder selection procedures based on AIC and BIC can be sub-optimal and introduce a method based on the focused information criterion (FIC) which targets the treatment effect by minimizing a prediction mean square error (see also the cross-validation method of \cite{brookhart2006semiparametric}). \cite{van2007super} introduces a \textit{Super Learner} estimator which is computed by selecting a candidate from a set of estimators obtained from different models using a cross-validation risk \citep{van2004cross, sinisi2007super}.

Bayesian adjustment for confounding (BAC) is a parametric variable selection approach introduced by \cite{wang2012bayesian}; BAC specifies a prior distribution for a set of possible models which includes a dependence parameter, $w\in [1,\infty]$, representing the odds of including a variable in the outcome model given that the same variable is in the propensity score model.  If we {\it know} a priori that a predictor of treatment is in fact a confounder, then $w$ can be set to $\infty$ \citep{crainiceanu2008adjustment, zigler2013model}. \cite{wilson2014confounder} proposes a decision-theoretic approach to confounder selection that can handle high-dimensional cases; a Bayesian regression model is fit and using the posterior credible region of the regression parameters, a set of candidate models is formed. A sparse model is then found by penalizing models that do not include confounders. This method is conservative in the sense that it may include instrumental variables that may inflate the variance of the treatment effect. Also, tuning the penalty function can be challenging.

Asymptotically, it is known that penalizing the conditional outcome model, given treatment and covariates, results in a valid variable selection strategy for causal effect estimation; however, for small to moderate sample sizes, it may result in the omission of weak confounders. The objective of this manuscript is to improve the small sample performance of the outcome penalization strategy while maintaining its asymptotic performance (Table \ref{tab:largep2}). We present a covariate selection procedure which facilitates the estimation of the treatment effect in the high-dimensional cases. Specifically, we propose a penalized  objective function which considers both covariate-treatment and covariate-outcome associations and has the ability to select even weak confounders. This objective function is used to identify the set of non-ignorable confounders and predictors of outcome; the resulting parameter estimates do not have any causal interpretation.  We derive the asymptotic properties of procedure and show that under some mild conditions the estimators have oracle properties (specifically, are consistent and asymptotically normally distributed).  We utilize the selected covariates to estimate the causal effect of interest using a doubly robust estimator.

\section{Preliminaries \& Notation} \label{sec:prelim}
In standard notation, let $Y(d)$ denote the (potential) outcome arising from treatment $d$, and let $D$ denote the treatment received. We consider for illustration the case of binary treatment.  The observed outcome, $Y$, is defined as $DY(1)+(1-D)Y(0)$.  We restrict attention here to the situation where each predictor can be classified into one of three types, and to single time-point studies.   We consider

\begin{itemize}

\item[(I)] \textit{treatment predictors} ($X_1$), which are related to treatment and not to outcome.

\item[(II)] \textit{confounders} ($X_2$), which are related to both outcome and treatment.

\item[(III)] \textit{outcome predictors} ($X_3$), which are related to outcome and not to treatment;
\end{itemize}
see the directed acyclic graph (DAG) in Figure \ref{fig:2conf}.
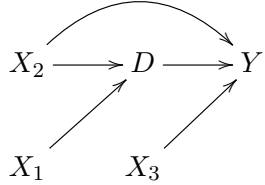
\begin{figure}[t]
\vspace{0.4 in}
\caption{ Covariate types: Type-I: $X_1$, Type-II: $X_2$ and Type-III: $X_3$.}
\vspace{0.4 in}
\centerline{
\xymatrix{
 X_2\ar[r]  \ar@/^2pc/[rr] & D \ar[r] & Y\\
  X_1 \ar[ur] & X_3 \ar[ur]\\
  }
}
\label{fig:2conf}
\end{figure}
In addition, as is usual, we will make the assumption of \textit{no unmeasured confounders}, that is, that treatment received $D$ and potential outcome to treatment $d$, $Y(d)$, are independent, given the measured predictors $X$.  In any practical situation, to facilitate causal inference, the analyst must make an assessment as to the structural nature of the relationships between the variables encoded by the DAG in Figure \ref{fig:2conf}.

\subsection{The Propensity Score for binary treatments}

The \emph{propensity score}, $\pi(.)$, for binary treatment $D$ is defined as $\pi(x) = \Pr(D=1|x)$, where $x$ is a $p$-dimensional vector of (all) covariates. In its random variable form, \cite{rosenbaum1983central} show that $\pi(X)$ is the coarsest function of covariates that exhibits the balancing property, that is, $D \perp X | \pi(X) $. As a consequence, the causal effect $\mu = \E[Y(1)-Y(0)]$ can be
computed by iterated expectation
\begin{equation}
\mu = \E_{X}[\E\{Y(1)|X\}-\E\{Y(0)|X\}] = \E_{\pi}[\E\{Y(1)|\pi\}-\E\{Y(0)|\pi\}],
\label{eq:ATE}
\end{equation}
where $\E_{\pi}$ denotes the expectation with respect to the
distribution of $\pi(X)$. For more details see \cite{rubin2008objective} and
\cite{rosenbaum2010causal}.


\medskip

\noindent \textbf{Remark 1:} In the standard formulation of the propensity score, no distinction is made between our three types of covariates.  Note that, however, for consistent estimation of $\mu$, \textit{it is not necessary to balance on covariates that are not confounders}. Covariates $X_1$ that predict $D$ but not $Y$ may be unbalanced in treated and untreated groups, but will not affect the estimation of the effect of $D$ on $Y$, as $D$ will be conditioned upon, thereby blocking any effect of $X_1$ \citep{de2011covariate}.  Covariates $X_3$ are unrelated to
$D$, so will by assumption be in balance in treated and untreated groups in the population. Therefore, the propensity score need only be constructed from confounding variables $X_2$.
Then, in the presence of outcome predictors $X_3$ of $Y$, the sequel to equation \eqref{eq:ATE}
takes the form
\begin{eqnarray}
\mu = \E[Y(1)-Y(0)] & = & \E_{X_2,X_3}[\E\{Y(1)|X_2,X_3\}-\E\{Y(0)|X_2,X_3\}] \nonumber\\[6pt]
& = & \E_{\pi_{2},X_3}[\E\{Y(1)|\pi_{2},X_3\}-\E\{Y(0)|\pi_{2},X_3\}].
\label{eq:ATE2}
\end{eqnarray}

\noindent \textbf{Remark 2:} Inclusion of covariates that are just related to the outcome in the propensity score model increases the covariance between the fitted $\pi$ and $Y$, and decreases the variance of the estimated causal effect, in line with the simulation of \cite{brookhart2006variable}.  


\subsection{Penalized Estimation}


In a given parametric model, if $\alpha$ is a $r$-dimensional regression coefficient,  $p_{\lambda}(.)$ is a penalty function and $l_m(\alpha)$ is the negative log-likelihood, the maximum penalized likelihood (MPL) estimator $\widehat \alpha_{ml}$ is defined as
\begin{align*}
\widehat \alpha_{ml}=\arg \min_{\alpha} \left [l_m(\alpha)+n\sum_{j=1}^r p_{\lambda}(|\alpha_j|) \right ].
\end{align*}
MPL estimators are shrinkage estimators, and as such, they typically have more finite sample bias, though less variation than unpenalized analogues. Commonly used penalty functions include LASSO \citep{tibshirani1996regression}, SCAD
\citep{fan2001variable}, Elastic Net (EN) \citep{zou2005regularization} and HARD \citep{antoniadis1997wavelets}. 

The remainder of this paper is organized as follows. Section \ref{sec:penalized} presents our two step variable selection and estimation procedure; we establish its theoretical properties.  The performance of the proposed method is studied via simulation in Section \ref{sec:simulation}. We analyze a real data set in Section \ref{sec:application}, and Section \ref{sec:conclude} contains concluding
remarks.  All proofs are given in the Appendix.


\section{Penalization and Treatment Effect Estimation} \label{sec:penalized}

In this section, we present our proposed method for estimating the treatment effect in high-dimensional cases.  We separate the covariate selection and treatment effect estimation procedure. First, we form a penalized objective function which is used to identify the important covariates, and establish the theoretical properties of the resulting estimators (i.e., minimizers of the penalized objective function). Note that because of the special characteristics of this function the estimators do not have any causal interpretation and are used just to {\it{prioritize}} variables. Second, treatment effect estimation is performed using a doubly robust estimator with the selected covariates.  We use a simple model structure to illustrate the methodology, and assume that a random sample of size $n$ of observations of outcome, exposure and covariates is available.

In order to present the method, we initially assume that columns of $\bX$ are orthogonal; this simplifying assumption is relaxed in Section \ref{sec:non-orth}.

\subsection{Penalized objective function }

Consider the following linear outcome model under a binary exposure $d$
\[
Y=\theta d+\bx\alpha_y+\epsilon,
\]
where $\bx$ is a $1 \times r$ { standardized } covariate vector, and $\epsilon$ is standard normal residual error.  Assuming a logit model for the propensity score, we have
\[
\pi(\bx,\alpha_d)=p(D=1|\bx,\alpha_d)=\frac{\exp\{\bx\alpha_d\}}{1+\exp\{\bx\alpha_d\}}.
\]
Let $\bX$ denote the corresponding $n \times r$ design matrix.  In the formulation, $\alpha_y$ and $\alpha_d$ denote $r \times 1$ vectors of parameters and $\theta$ is the treatment effect.

First,  we form an objective function, $M(\alpha)$, in which the coefficients $\alpha_y$ and $\alpha_d$ are replaced by a single common vector of coefficients $\alpha$ that is proportional to a weighted sum of $|\alpha_y|$ and $|\alpha_d|$, where $|.|$ denotes the componentwise absolute value. This will guarantee that the objective function satisfies the following condition:

\medskip
\noindent {\it{$Argmin$ Condition:}} An element, $\hat \alpha_j$, of the minimizer $\hat \alpha$ of $M(\alpha) $ converges to zero as $n \rightarrow \infty$ if the corresponding covariate is not associated with either treatment or outcome.


\medskip

In standard linear regression where we regress $Y$ on a vector of covariates $\bX$ (with no treatment variable) that are presumed orthogonal, an example of an objective function that estimates $|\alpha_y|$ is
\[
M(\alpha) = \frac{1}{2n}\left(\left|\sum_{i=1}^n \bx_i^\top y_i \right| - n\alpha \right)^\top \left(\left|\sum_{i=1}^n \bx_i^\top y_i \right| - n\alpha \right)
\]
yielding the estimating equation and estimate
\[
\dfrac{\partial M(\alpha) }{\partial \alpha} = \left|\sum_{i=1}^n \bx_i^\top y_i \right|  - n\alpha = 0
\qquad \Longrightarrow \qquad
\widehat \alpha = \frac{1}{n} \left|\sum\limits_{i=1}^n \bx_i^\top y_i\right|.
\]
We have used the fact $n^{-1} \bx^\top \bx$ is  an identity matrix due to standardization. It is clear that as $n \longrightarrow \infty$, 
$\widehat \alpha \stackrel{p}{\longrightarrow} |\E[x_i^\top Y_i]|$
Similarly, in a binary exposure setting, $|\alpha_d|$ can be estimated using the objective function $ [-\left|\sum_{i=1}^n \bx_id_i\right|\alpha +\sum_{i=1}^n \log(1+\exp\{ \bx_i \alpha \})]$.  Now, we combine these two objective functions and show that the resulting objective function have some interesting features. First, we obtain the least squares (or ridge) estimate $\widetilde \theta$ of $\theta$ by regressing  $Y$ on $ D$ and the vector of covariates $\bX$. Because of the inclusion of spurious and instrumental variables in the model, $\widetilde \theta$  is not efficient   but is consistent for $\theta$ under the no unmeasured confounders assumption. Define $\widetilde Y=Y-\widetilde \theta D$.  Let
\begin{align}
M(\alpha)=\frac{1}{2} \Big[\Big| \sum_{i=1}^n\bx_i \widetilde y_i\Big|- \alpha^\top  \sum_{i=1}^n(\bx_i^\top \bx_i)  \Big]  &\left(\sum_{i=1}^n\bx_i^\top \bx_i\right)^{-1} \Big[\Big| \sum_{i=1}^n\bx_i \widetilde y_i\Big|-     \alpha^\top  \sum_{i=1}^n(\bx_i^\top \bx_i) \Big]^\top+  \nonumber \\
& \frac{1}{{ \tau}} \Big[-\Big| \sum_{i=1}^n\bx_id_i\Big|\alpha + \sum_{i=1}^n\log(1+\exp\{ \bx_i \alpha \})\Big],
\label{eq:misobj}
\end{align}
{ where $\tau$ is a positive constant. Under orthogonality and standardization of $\bX$, we have
\begin{align*}
M(\alpha)=\frac{1}{2n} \Big[\Big| \sum_{i=1}^n\bx_i \widetilde y_i\Big|- n\alpha^\top  \Big]  & \Big[\Big| \sum_{i=1}^n\bx_i \widetilde y_i\Big|-     n\alpha^\top   \Big]^\top+  \nonumber \\
&  \frac{1}{\tau} \Big[-\Big| \sum_{i=1}^n\bx_id_i\Big|\alpha + \sum_{i=1}^n\log(1+\exp\{ \bx_i \alpha \})\Big]
\end{align*}
}
Note that for each $j$, the parameter $\alpha_j$ corresponding to $x_j$ is the same in both models. Now, we show that in this function the absolute values play a critical rule in satisfying the $Argmin$ condition.  Let $\hat \alpha = \arg \min_{\alpha} M(\alpha)$ and $\tildealy$ be the least squares estimate of the parameters in the outcome model. By convexity of $M(\alpha)$, $\hat \alpha$ must be a solution to $\partial M(\alpha)/\partial \alpha=0$, which implies that {
\begin{align}
\left[ n \widehat \alpha + \frac{1}{\tau} \sum_{i=1}^n \bx_i^\top \frac{\exp\{\bx_i\widehat \alpha \}}{1+\exp\{\bx_i\hat \alpha \}} \right] &=   \left|\sum_{i=1}^n \bx_i^\top \widetilde y_i\right| + \frac{1}{\tau}\left|\sum_{i=1}^n \bx_i^\top d_i\right|  =n \left|\tildealy\right| + \frac{1}{\tau}\left|\sum_{i=1}^n \bx_i^\top d_i\right| .
  \label{eq:alpha}
\end{align} }
This equation shows that $\hat \alpha=0$ (i.e. the $Argmin$ condition is satisfied) if $cov(X,Y)=cov(X,D)=\textbf{0}$. Note that $cov(X,D)=\textbf{0}$ implies  $cov(X,\widetilde Y)=cov(X,Y)$. { Moreover, using the first three terms of a Taylor expansion, we have
\[
\widehat \alpha = \frac{2\tau}{2\tau+1} \left|\tildealy\right| +\frac{1}{2\tau+1} \left|\tildeald \right|.
\]
Thus, $\hat \alpha$ is a weighted sum of $|\tildealy|$ and $|\tildeald|$ where $\tildeald$ is the  maximum likelihood estimate of the parameters in the treatment model. The constant $\tau$ controls the contribution of  components $\left|\tildealy\right|$ and $\left|\tildeald\right|$ to the estimator. For example, for $\tau=2$,  $\widehat \alpha = \frac{4}{5} \left|\tildealy\right| +\frac{1}{5} \left|\tildeald \right|$; while for $\tau=0.1$, $\widehat \alpha = \frac{0.2}{1.2} \left|\tildealy\right| +\frac{1}{1.2} \left|\tildeald \right|$. Thus, $\tau$ gives a flexibility to our objective function such that as it decreases to zero the proposed estimate $\widehat \alpha$ converges to $\left|\tildeald \right|$. See Section \ref{sec:tau}.  }

Figure \ref{fig:mislike-perf} visually presents how $\widehat \alpha_j$ behaves, for a fixed $\tau=0.5$, when $\widehat \alpha_{jy}$ converges to zero as sample size increases, i.e., the $j$th variable becomes insignificant.  Specifically, this figure presents a case where there is just one covariate (i.e., $j=1$) and the coefficient of this covariate in outcome and treatment models are $\alpha_{jy}=1/\sqrt n$ and $\alpha_{jd}=0.3$, respectively, where $n$ is the sample size. As expected, $\hat \alpha$ corresponding to this covariate does not converge to zero as sample size increases. The same behaviour would be observed if $\alpha_{jd} \rightarrow 0$ as $n \rightarrow \infty$ and $\alpha_{jy}$ is a non-zero constant.

\begin{figure}[t]
\centering
\makebox{\includegraphics[scale=.50]{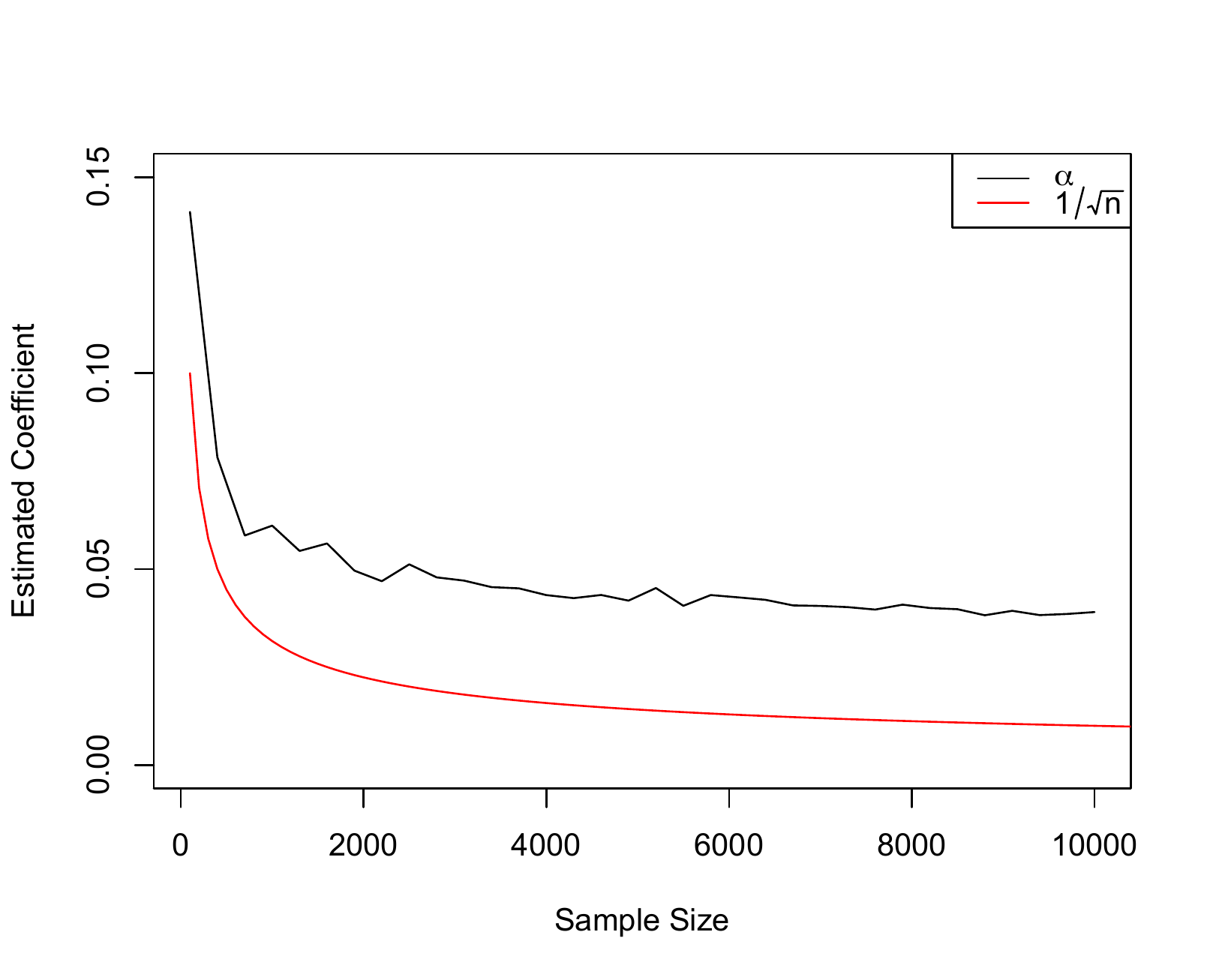}}
\caption{ { Performance of the modified  objective function based estimator for different sample sizes $n$. Red and black  lines are $1/\sqrt n$ and the estimated coefficient $\widehat \alpha$ using the modified  objective function.  }}
\label{fig:mislike-perf}
\end{figure}

\noindent Therefore,  penalizing  objective function (\ref{eq:misobj}) results in selecting covariates that are either related to outcome (Type-III) or treatment (Type-I). However, this is against our goal of keeping variables that are either predictors of the outcome or non-ignorable confounders and excluding instrumental variables (Type-I). To deal with this problem, we present a weighted lasso penalty function that is tailored specifically for causal inference variable selection:
\[
\frac{\lambda}{( \tildealy)^2(1+| \tildeald|)^2}  \sum_{j=1}^r |\alpha_j|,
\]
where $\tildealy$ and $\tildeald$ are the least squares (or ridge) and maximum likelihood estimates of the parameters in the outcome and treatment models, respectively, and $\lambda$ is a tuning parameter.  Thus, the proposed modified  penalized objective function is given by
\begin{align}
M_p(\alpha)=M(\alpha)
+ \lambda \sum_{j=1}^r  \frac{|\alpha_j|}{( \widetilde \alpha_{jy})^2(1+| \widetilde \alpha_{jd}|)^2}
\label{eq:peemis}
\end{align}
We refer to $ \widehat\alpha=\arg \min_{\alpha} M_p(\alpha)$ as  penalized modified objective function estimators (PMOE). The magnitude of the penalty on each parameter is proportional to its contribution to the outcome and treatment model. Note that as $\hat \alpha_{jy} \rightarrow 0$, our penalty function puts more penalty on the $j$th parameter while considering the covariate-treatment association. For example, when a covariate barely predicts the outcome and treatment, our proposed penalty function imposes a stronger penalty on the parameter compared to a case where a covariate barely predicts the outcome and is strongly related to treatment. This is an important feature of the proposed penalty function which allows selecting such weak confounders for small sample sizes.  The proposed penalized estimator asymptotically selects the same covariates as the adaptive lasso on the outcome model \citep{zou2006adaptive}. Thus, the proposed method improves the small sample performance of  the outcome model penalization strategy while maintaining the asymptotic properties of this strategy.

\subsection{The role of $\tau$ in PMOE } \label{sec:tau}
{
The constant $\tau$ reflects investigators belief about the importance of including variables that are weakly(strongly) related to outcome and strongly(weakly) related to treatment, i.e., for smaller values of $\tau$, variables that are weakly related to the outcome but strongly to the treatment have more contribution to $\widehat \alpha$ and have more chance to be selected. Also, when $\tau=\infty$, the treatment mechanism does not have any contribution in the objective function (\ref{eq:misobj}). Thus, our procedure performs similarly to the outcome model based variable selection that may exclude non-ignorable confounders that are weakly related to the outcome.  Our simulation studies shed more light on the role of $\tau$ in our procedure.
}

\subsection{Main Theorem} \label{theorems}


Suppose $\alpha_0=(\alpha_{01},\alpha_{02})$ is the true parameter value of the  $r$-dimensional vector of parameters where
$\alpha_{02}=\{\alpha_j, j=s+1,...,r\} \equiv \mathbf{0}$ contains those elements of $\alpha$ that are in fact zero, so that the corresponding predictors are not confounders; $s$ denotes the true number of predictors present in the model ({\it exact sparsity} assumption).    Let $\widehat \alpha=(\widehat \alpha_1,\widehat \alpha_2)$ be the vector of estimators  corresponding to \eqref{eq:peemis}. The next theorem proves the sparsity and asymptotic normality of  the proposed penalized estimators under the following two conditions:
\begin{enumerate}
\item[] (a) Let
\[
\sum_{i=1}^n \epsilon_{i\widetilde y}= |\sum_{i=1}^n\bx_i\widetilde y_i|- n\alpha_0^\top \qquad \text{and} \qquad \sum_{i=1}^n \epsilon_{id}= |\sum_{i=1}^n d_i\bx_i|-\sum_{i=1}^n \bx_i\frac{\exp\{\bx_i\alpha_0 \}}{1+\exp\{\bx_i\alpha_0 \}}
\]
and $e_i=\epsilon_{i\widetilde y}+\frac{1}{\tau}\epsilon_{id}$. We assume that $ \frac{ 1}{\sqrt n}  \sum_{i=1}^n [\epsilon_{i\widetilde y}+\frac{1}{\tau}\epsilon_{id}]$ converges in distribution  a multivariate normal  $N_{r}(0,\Sigma(\alpha_0))$. 
\item[] (b) Let
\[
 \frac{1}{n} \left[n \I+ \frac{1}{\tau} \sum_{i=1}^n \frac{\exp\{\bx_i \alpha_0\}}{[1+\exp\{\bx_i \alpha_0\}]^2} \bx_i^\top \bx_i  \right] \stackrel{p}{\longrightarrow} \Omega(\alpha_0),
 \]

where $\Omega(\alpha_0)$ is a  $r\times r$ positive definite matrix and $\I$ is the identity matrix.
\end{enumerate}

\medskip
\begin{theorem} \textbf{(Oracle properties)}
Suppose conditions (a) \& (b) are fulfilled, further $\lambda_n/\sqrt n \rightarrow 0$ and $\lambda_n \sqrt n \rightarrow \infty$.  Then
\begin{enumerate}
\item[(a)] $\Pr(\widehat \alpha_{2}= \mathbf{0})\rightarrow 1$ as $n\rightarrow \infty$
\label{th:norm}
\item[(b)] $\sqrt n(\widehat \alpha_{01}- \alpha_{01}) \CiD N(0,\Omega_{11}^{-1}\Sigma_{11}\Omega_{11}^{-1}),$
\end{enumerate}
where $\alpha_{01}=\alpha_{01}$ is the true vector of
non-zero coefficients. Also, $\Omega_{11}$ and $\Sigma_{11}$ are corresponding elements of $\Omega$ and $\Sigma$, respectively.
\label{th:oracle}
\end{theorem}

\subsection{Choosing the Tuning Parameter} \label{tuning parameter}

We select the tuning parameter using the {\it Generalized Cross Validation}
(GCV) method suggested by \cite{tibshirani1996regression} and
\cite{fan2001variable}. Let $\bX_{\lambda}$ be the selected covariates corresponding to a specific value of $\lambda$, we first regress $\widetilde Y$ on $\bX_{\lambda}$ and calculate the residual sum of square (RSS) of this model. Then
\[
\text{GCV}(\lambda)=\frac{ \text{RSS}(\lambda)/n}{[1-d(\lambda)/n]^2},
\]
where $d(\lambda) =
\text{trace} [\bX_{\lambda}(\bX_{\lambda}^\top \bX_{\lambda} + n\Sigma_{\lambda}(\widehat
\alpha))^{-1})\bX_{\lambda}^\top]$ is the effective number of parameters and
$\Sigma_{\lambda}( \alpha) =
\text{diag}[p^\prime_{\lambda}(|\alpha_1|)/|\alpha_1|,...,p^\prime_{\lambda}(|\alpha_{r}|)/|\alpha_{r}|]$.   The selected tuning parameter $\widehat \lambda$ is defined by $\widehat
\lambda= \text{arg}\min_{\lambda} \text{GCV}(\lambda)$.

\subsection{Estimation of the causal effect} \label{sec:est}

For treatment effect estimation, we fit the following model using the set of covariates selected in the previous step; note that a user may want to use other causal adjustment models such as inverse probability weighting or propensity score matching, and the selection approach can be used for these procedures also.

Our model is a slight modification of the conventional propensity score regression approach of \cite{RobinsMarkNewey1992}, and specifies
\begin{align}
\E[Y_i|S_i=s_i, \bX_i=\bx_i]=\theta s_i+ g(\bx;\gamma), \label{eq:PSR5}
\end{align}
where $S_i=D_i-\E[D_i|\bx_i] = D_i -\pi(\bx_i)$, $g(\bx;\gamma)$ is a function of
covariates and $\pi$ is the propensity score. The quantity $S_i$ is used in
place of $D_i$; if $D_i$ is used the fitted model may result in a biased estimator for $\theta$ since $g(\bx;\gamma)$ may be incorrectly specified. By defining $S_i$ in this way, we restore $ cor(S_i,\bX_{ij})=0$ for
$j=1,2,..,p$ where $p$ is the number of  selected variables (if $\pi(\bx_i) = \E[D_i|\bx_i]$ is correctly specified), as $\pi(\bx_i)$ is the (fitted) expected value of $D_i$, and hence $\bx_j^\top (D - \pi(\bx) )=0$, where $\bx_j^\top = (x_{1j},\ldots,x_{nj})$. Therefore,
misspecification of $g(.)$ will not result in an inconsistent estimator of $\theta$.

In general, this model results in a \textit{doubly robust} estimator (see \cite{davidian2005semiparametric}, \cite{Schafer2005} and \cite{bang2005doubly}); it yields a consistent estimator of $\theta$ if \textit{either} the propensity score model or conditional mean model \eqref{eq:PSR5} is correctly specified, and is the most efficient estimator \citep{tsiatis2006semiparametric} when both are correctly specified. For additional details on the related asymptotic and finite sample behavior, see \cite{kang2007demystifying}, \cite{neugebauer2005prefer}, \cite{van2003unified}
and \cite{robins1999robust}. \\

\noindent \textbf{Remark 3:} The importance of the doubly robust estimator is that, provided the postulated outcome and treatment model identify the true non-zero coefficients in each model, the proposed estimation procedure will consistently estimate the treatment effect given that at least one of the propensity score or outcome models are correctly specified. Assuming linear working models, a sufficient but {\it{not}} necessary condition for selecting non-ignorable confounders is the linearity of the true models in their parameters. 

The model chosen for estimation of the treatment effect is data dependent. Owing to the inherited uncertainty in the selected model, making statistical inference about the treatment effect becomes ``post-selection inference". Hence, inference about the treatment effect obtained in the estimation step needs to be done with caution. The weak consistency of the post-selection estimator results from the following theorem.
\begin{theorem}
Let $\zeta(\hat \theta_{M_n},M_n)$ be a smooth function of $\hat \theta_{M_n}$ where $M_n$ denotes a model formed from a set of selected variables using our method. Then $\zeta(\hat\theta_{M_n},M_n) \CiP \zeta(\theta_{M_0},M_0)$ as $n \rightarrow \infty$ where $M_0$ denotes a model formed from the set of non-zero coefficients.
\label{th:consis}
\end{theorem}

\subsection{The Procedure Summary}

The penalized treatment effect estimation process explained in Sections 3.1 to \ref{sec:est} can be summarized as follows: \vspace{-.1in}
\begin{enumerate}
\item Estimate the vector of parameters $\hat \alpha=\arg\min_{\alpha} M_p(\alpha)$ where $M_p(\alpha)$ is defined in (\ref{eq:peemis}).
\item Using the covariates with $\alpha \neq 0$, estimate the propensity score $\hat \pi(\bX)$.
\item Define a random variable $S_i=D_i-\hat\pi(\bX_i)$ and fit the outcome model $\E[Y_i|d,\bx]=\theta s_i+ g(\bx_i;\gamma).$
The vector of parameters $(\theta,\gamma)$ is estimated using ordinary least squares. For simplicity, we assume the linear working model for $g(\bx_i;\gamma)=\gamma^\top\bx_i$. The design matrix $\bX$ includes a subset of variables with  $\alpha \neq 0$.
\end{enumerate}

\subsection{Non-orthogonal Covariates}\label{sec:non-orth}

In this section we relax the assumption of orthogonality of the covariates. Let $\bX$ denote a $n \times r$ design matrix such that the covariance matrix $var(\bX)$ has some non-zero off-diagonal entries. Let $\bU$ be the orthogonalized matrix of covariates constructed using orthogonalization techniques such as the Gram-Schmidt algorithm. Proposition \ref{co:consis} states certain facts about the performance of the proposed method when the original vector of covariates $\bX$ is replaced by $\bU$ in the objective function $M(.)$. Specifically, it states that the proposed method still has the sparsity property as long as the penalty function is constructed based on the original $\bX$, i.e., $\tildealy$ and $\tildeald$ are the least squares (or ridge) estimate of the parameters in the outcome and treatment models given $\bX$, respectively. Define


\begin{align}
M^u(\alpha)=\frac{1}{2n} \Big[\Big| \sum_{i=1}^n\bu_i \widetilde y_i\Big|- n\alpha^\top  \Big]  & \Big[\Big| \sum_{i=1}^n\bu_i \widetilde y_i\Big|-     n\alpha^\top   \Big]^\top+  \nonumber \\
&  \frac{1}{\tau} \Big[-\Big| \sum_{i=1}^n\bu_id_i\Big|\alpha + \sum_{i=1}^n\log(1+\exp\{ \bu_i \alpha \})\Big]
\label{eq:objnon-orth}
\end{align}
 Similar to Section \ref{theorems}, let $ \alpha_0^u=\alpha_{01}^u$ be the true parameter value in the model that includes $\bU$. Let $\widehat \alpha^u$ be the solution to
\[
\arg \min_{(\alpha^u)} M^u(\alpha^u)+ \lambda \sum_{j=1}^r  \frac{|\alpha_j^u|}{( \widetilde \alpha_{jy})^2(1+| \widetilde \alpha_{jd}|)^2}.
\]
 Because of the orthogonalization there might be mismatches between the indices of coefficients with zero value. For example, it may happen that the $j$th element of the vector of parameters $\alpha_j^u \neq 0$ while $\alpha_j = 0$. Therefore, if our goal is to identify the set of zero coefficients based on the covariate vector $\bX$, the penalty function has to impose heavier penalty on coefficients that are significant in a model based on $\bU$ but insignificant  in the original covariate vector $\bX$. Similarly, it may happen that the $j$th element of the vector of parameters $\alpha_j^u = 0$ while $\alpha_j \neq 0$. In this case, the penalty imposed on the $j$th covariate has to go to zero as $n$ goes to infinity. This ensures that the penalty is not strong enough to set $\hat \alpha_j^u = 0$ for any finite $n$.  


\begin{proposition}
 Let $\lambda_n/ \sqrt n \rightarrow 0$ and $\lambda_n \rightarrow \infty$.  Then $\Pr(\widehat \alpha^u_j=0| \alpha_j=0 )\rightarrow 1$ as $n\rightarrow \infty$ and, for any finite $n$, if $\alpha_j \neq 0$,  $\widehat \alpha^u_{j} \neq 0$.
\label{co:consis}
\end{proposition}
The results of this proposition relies on the choice of the penalty function. Because the penalty function is defined based on $\tildealy$ and $\tildeald$; that are the least squares (or ridge) estimate of the parameters given $\bX$, the strength of the imposed penalty  is derived by  $\bX$ and not $\bU$. This helps us, first, to estimate $\widehat \alpha_j^u=0$ when $ \alpha_j=0$ even if $ \alpha_{0j}^u\neq 0$, and second, estimate $\widehat \alpha_j^u \neq 0$ when $ \alpha_j \neq 0$ even if $ \alpha_{0j}^u= 0$.

\section{Simulation Studies} \label{sec:simulation}
In this section, we study the performance of our proposed variable selection method using  simulated data. We compare our results with BAC method introduced by \cite{wang2012bayesian},  the Bayesian credible region (Cred. Reg.) introduced by \cite{wilson2014confounder} and outcome penalized estimator (Y-fit). Our simulation  includes a scenario in which there is a weak confounder that is strongly related to the treatment but weakly to the outcome. We consider linear working models for  $g()$ throughout this section.

\begin{table}[t]
\caption{ Simulation results for scenarios 1 \& 2. Bias, S.D. and MSE are for the treatment effect. S.D: empirical standard error. Y-fit is obtained by penalizing the outcome model via the LASSO penalty.}
\label{tab:largep}
\begin{center}
\begin{tabular}{lrrr|rrrrr|} \hline

Method   &\multicolumn{1}{c}{Bias} & \multicolumn{1}{c}{S.D}  & \multicolumn{1}{c|}{MSE} &
\multicolumn{1}{c}{Bias} & \multicolumn{1}{c}{S.D} & \multicolumn{1}{c}{MSE}  \\ \hline
Scenario 1.               & \multicolumn{3}{c}{$n=300$} & \multicolumn{3}{c}{$n=500$} \\
PMOE$^{\tau=0.1}$      &    0.031 & 0.633 & 0.402&0.036 &0.431 & 0.187          \\
PMOE$^{\tau=0.5}$      &    0.004 & 0.580 & 0.337&0.041 &0.432 & 0.188          \\
PMOE$^{\tau=1}$       &    0.007 &0.598  &0.357&0.040 & 0.431  &0.188         \\
PMOE$^{\tau=20}$       &    0.008 &0.598 & 0.358&0.040 & 0.431  &0.188         \\
Cred. Reg.  & 0.004 &0.780 &0.608&0.024 &0.585 &0.342  \\
BAC ($\omega$=$\infty$)  &          0.064 &1.332&1.781&0.040&1.062 &1.130  \\
Oracle   &          0.005 &0.593 & 0.351&0.030& 0.430 & 0.186 \\ \hline
Scenario 2.               & \multicolumn{3}{c}{$n=300$} & \multicolumn{3}{c}{$n=500$} \\
PMOE$^{\tau=0.1}$       &   0.027 &0.556  &0.310&0.008 &0.419  &0.175        \\
PMOE$^{\tau=0.5}$       &    0.061 &0.561  &0.321&0.032 &0.415 & 0.173        \\
PMOE$^{\tau=1}$       &    0.123 &0.587  &0.360&0.065 & 0.424  &0.184        \\
PMOE$^{\tau=20}$       &    0.675 &0.584  &0.798&0.788 &0.465  &0.837       \\
Cred. Reg.  &0.010 &0.765 & 0.582&0.027 &0.595  &0.355  \\
Y-fit   &         0.710 &0.598  &0.862&0.818&0.453  &0.875 \\
BAC ($\omega$=$\infty$)  &          0.070 &1.242 &1.542&0.008&1.011 &1.022  \\
Oracle   &         0.026 &0.560 & 0.315&0.025& 0.419 & 0.176 \\ \hline
\end{tabular}
\end{center}
\end{table}

We generate 500 data sets of sizes 300 and 500 from the following two models:
\begin{enumerate}
\item[1.] $D \sim \text{Bernoulli}\left(\dfrac{\exp\{0.2x_1-2x_2+x_5-x_6+x_7-x_8
\}}{1+\exp\{0.2x_1-2x_2+x_5-x_6+x_7-x_8\}}\right)$ \\[12pt]
$Y \sim \text{Normal}(d+2x_1+0.5x_2+5x_3+5x_4, 4)$

\bigskip

\item[2.] $D \sim \text{Bernoulli}\left(\dfrac{\exp\{0.2x_1-2x_2+x_5-x_6+x_7-x_8}
    {1+\exp\{0.2x_1-2x_2+x_5-x_6+x_7-x_8\}}\right)$, \\[12pt]
$Y \sim \text{Normal}(d+2x_1+0.2x_2+5x_3+5x_4, 4)$
\end{enumerate}
where $\bX_k$ has a $N(1,4)$ for $k=1,...,100$. Note that in the second scenario, $x_2$ is considered as a weak confounder.  Results are summarized in Table \ref{tab:largep}; the Y-fit row refers to the estimator obtained by penalizing the outcome model using {\it LASSO} penalty, and the {\it{Oracle}} row refers to estimator obtained by including $(x_1,x_2,x_3,x_4)$ in the propensity score and outcome model (\ref{eq:PSR5}). We studies the performance of PMOE for different values of constant $\tau=0.1,0.5,1$ and 20. In the first scenario  there is no weak confounder and the  Y-fit is omitted since the results are similar to the PMOE$^{\tau=20}$   row.  The Bayesian adjustment for confounding (BAC) has been implemented using the \texttt{R} package \texttt{BEAU} with $\omega=\infty$ and the Bayesian credible region (Cred. Reg.) has been implemented using the \texttt{R} package \texttt{BayesPen} with flat prior.

{ The variance of the estimator in the BAC is too large due to the inclusion of  spurious variables that are not related to the outcome. The PMOE   and Creg. Reg. estimators, however, are unbiased and have smaller variance. In fact, PMOE has the lowest variance compared to BAC and Creg. Reg. methods regardless of the value of $\tau$.  In the second scenario, the Y-fit estimator is bias because of under selecting the confounder $X_2$. Also, the bias of the proposed PMOE estimator increases by increasing $\tau$ that is expected because as $\tau$ increases the proposed method should perform similarly to the outcome based variable selection methods such as Y-fit. 

Table \ref{tab:largep2} presents the average number of coefficients set to zero correctly and incorrectly under the second scenario. There are four non-zero coefficients in our generative model so the number in the correct column should be 96 and in the incorrect column should be 0. This table shows that both Creg. Reg. and BAC are somewhat conservative and include some of the variables that should not be included which can be the source of the observed variance inflation in Table \ref{tab:largep}. This is mostly due to inclusion of variables that are predictor of treatment model but have no association with the outcome. Also, increasing $\tau$ in PMOE, increases the chance of setting the coefficient of $X_2$ to zero. The Y-fit row shows that this method is setting a nonzero coefficient to zero (i.e., coefficient of $X_2$) which explains the bias in Table \ref{tab:largep}. This, in fact, highlights the importance of our proposed method. Our simulation studies suggest that we should avoid large values of $\tau$ if we are concern about excluding weak confounders such as $X_2$.

}

\begin{table}[t]
\caption{ { { Simulation results for scenarios 1 \& 2. Number of coefficients that are correctly or incorrectly set to zero. Y-fit is obtained by penalizing the outcome model via LASSO penalty.  }}}\centering
\begin{tabular}{c c c  |c c c c c c} \hline
Method  & \multicolumn{1}{c}{Correct} & \multicolumn{1}{c}{Incorrect}
 & \multicolumn{1}{c}{Correct} & \multicolumn{1}{c}{Incorrect}\\
\hline
  & \multicolumn{2}{c}{$n=300$} & \multicolumn{2}{c}{$n=500$}\\
PMOE$^{\tau=0.1}$   &95.40 &0.03   &  95.74 &0.01   \\
PMOE$^{\tau=0.5}$   &95.90 &0.09   &  96.00 &0.01   \\
PMOE$^{\tau=1}$   &96.00 &0.18   &  96.00 &0.06   \\
PMOE$^{\tau=20}$   &96.00 &0.76   &  96.00 &0.89  \\
Cred. Reg.   &92.39 &0.00  &   92.61 &0.00    \\
BAC ($\omega$=$\infty$)  &91 &0.00  &   91 &0.00    \\
 Y-fit      &97 &0.90   &   97 &0.92   \\           \hline
\end{tabular}
\label{tab:largep2}
\end{table}

Simulation studies presented in Appendix 2 study the performance of our  covariate selection and estimation procedure when the covariates are non-orthogonal. Tables 5 \& 6 show that the proposed method is still outperforming the other methods. Moreover, in Appendix 3, we study  cases that the number of covariates is larger than the sample size ($r>n$). We also investigate cases where either of the working models of the propensity score or the outcome model is misspecified. Our results show that the proposed method performs well, and outperforms Y-fit

\section{Application: the effect of life expectancy on economic growth} \label{sec:application}

In this section we examine the performance of our proposed method on the cross-country economic growth data used by \cite{doppelhofer2003determinants}. For illustration purposes, we focus on a subset of the data which includes 88 countries and 35 variables. Additional details are provided in \cite{doppelhofer2009jointness}. We are interested in selecting non-ignorable variables which confound the effect of {\it life expectancy} (exposure variable) as a measure of population health on the {\it average growth rate of gross domestic product per capita in 1960-1996} (outcome).

The causal (or, at least, \emph{unconfounded}) effect of life expectancy on economic growth is controversial. \cite{acemoglu2006disease} find no evidence of increasing life expectancy on economic growth while \cite{husain2012alternative} shows that it might have positive effect. We dichotomize the life expectancy based on the observed median, which is 50 years.  Hence, the exposure  variable D=1 if life expectancy is below 50 years in that country and 0 otherwise.



\begin{table}
\caption{The economic growth data: List of significant variables. Y-fit is obtained by penalizing the outcome model via LASSO penalty. }\centering
\begin{tabular}{| l |rrrrrr|} \hline
  Variable   & Y-fit&BAC & Cred. Reg.& PMOE& PMOE& PMOE  \\ 
     & &$\omega$=$\infty$& & $\tau=0.1$& ${\tau=0.5}$& ${\tau=5}$  \\ 
  \hline
 Air Distance to Big Cities &---&---& ---& $\surd$& $\surd$& $\surd$ \\
 Ethnolinguistic Fractionalization &$\surd$&$\surd$&$\surd$& $\surd$& $\surd$& $\surd$\\
Fraction of Catholics&--- &---&---& $\surd$& $\surd$& $\surd$\\
Population Density 1960 &--- &---&$\surd$& $\surd$& $\surd$& ---\\
 East Asian Dummy &$\surd$ &$\surd$&$\surd$& $\surd$&$\surd$& $\surd$\\    
 Initial Income (Log GDP in 1960)    &--- &$\surd$&---& $\surd$&$\surd$& $\surd$\\
Public Education Spending Share    &$\surd$ &---&$\surd$&---&---&---\\
Nominal Government Share    &--- &---&---&$\surd$&---&---  \\
Investment Price    &$\surd$ &$\surd$&$\surd$& ---&---&--- \\
Land Area Near Navigable Water    &--- &---&$\surd$& ---&---&---\\
Fraction GDP in Mining    &--- &---&$\surd$&$\surd$&$\surd$&$\surd$\\
 Fraction Muslim   &--- &---&$\surd$&---& ---&---\\
 Political Rights    &--- &$\surd$&---&$\surd$&$\surd$&$\surd$\\
 Real Exchange Rate Distortions   &$\surd$ &---&---& ---&---&---\\
 Colony Dummy  &--- &---&$\surd$&---&---&$\surd$\\
 European Dummy    &$\surd$ &---&---& ---&---&---\\
 Latin American Dummy   &$\surd$ &$\surd$&---&$\surd$ &$\surd$&$\surd$\\
 Landlocked Country Dummy   &--- &---&$\surd$&---&---&---\\ 
 Oil producing Country Dummy   &--- &---&---&$\surd$&---&---\\ 
Land Area Near Navigable Water&--- &---&$\surd$&$\surd$&---&---\\ \hline
\end{tabular}
\label{tab:list}
\end{table}

We select the significant covariates  for the conditional mean and the treatment  models using the penalized objective function (\ref{eq:peemis}).  After covariate selection, we fit the model $\E[Y|s,\bx] = \theta s+ g(\bx;\gamma)$, where
$\theta$ is the treatment effect parameter (the function $g()$ assumed to be linear). Interaction or the higher order of the propensity score can be added to the outcome model if needed.

In our analysis, Y-fit refers to the case where just the outcome model is penalized using LASSO to select the significant covariates. We also implement the BAC with $\omega=\infty$ and Credible region with flat prior  methods. Table \ref{tab:list} presents the list of variables and their estimated coefficients which are selected at least by one of the methods.

{ The proposed method selects 12 and 9 variables depending on the value of $\tau$ while Y-fit, BAC and Cred. Reg. select 7, 6 and 11 variables, respectively.  Y-fit and PMOE$^{\tau=5}$ are under selecting some of non-ignorable confounders which   barely predict the outcome.  Specifically, {\it Population Density 1960}, and {\it Initial Income } are such non-ignorable confounders which are known to be significant in the economics literature. Table \ref{tab:gdata} shows that this under selection leads to a biased treatment effect estimator. BAC and Creg. Reg. are also ignoring some of the important potential confounders. Specifically, BAC does not select {\it{Population Density 1960}} and Creg. Reg. is ignoring {\it{Initial outcome in 1960}}. \cite{doppelhofer2011robust} listed the latter two variables as potential confounders.  These under selections may due the non-linear nature of the outcome or/and treatment model. Moreover, because there are strong evidence that variables such as {\it{Oil producing Country, Land Area Near Navigable Water, and Nominal Government Share}} are unlikely to confound the effect of the life expectancy on the economic growth \citep{doppelhofer2009jointness, acemoglu2006disease, ley2007jointness, ley2009comments, ley2009effect, magnus2010comparison, doppelhofer2011robust, eicher2011default},   PMOE$^{\tau=0.5}$ seems to select more reasonable covariates than PMOE$^{\tau=0.1}$.

To gain insight into the effect of parameter $\tau$ on the selection results, we plot the estimated coefficients $\alpha$ for different values of $\tau$ for given tuning parameters. In Figure \ref{fig:difftau}, the blue solid (dark dashed) lines correspond to coefficients that their estimated value is (not) greater than 0.05 for the entire range of $\tau$. Also, displays  (a)--(d) correspond to tuning parameter values $\lambda=0.0$, 0.0001, 0.001, and 0.01, respectively.  Figure \ref{fig:difftau} shows that, for different values of $\tau$, the selected set of covariates may vary slightly. For example, in Figure \ref{fig:difftau}(c) where a moderate penalty function is imposed, the proposed method suggests to include  {\it{Colony}} (the dashed line)  to the selected set for larger values of $\tau$. Also, as the penalty becomes stronger, the estimated coeffiecints become more stable across values of $\tau$ (Figure \ref{fig:difftau}(d)). This is because, for larger values of $\lambda$, the penalty function has a more dominant rule on the variable selection than $\tau$. 
}

\begin{figure}[t]
\centering
\makebox{\includegraphics[scale=.7]{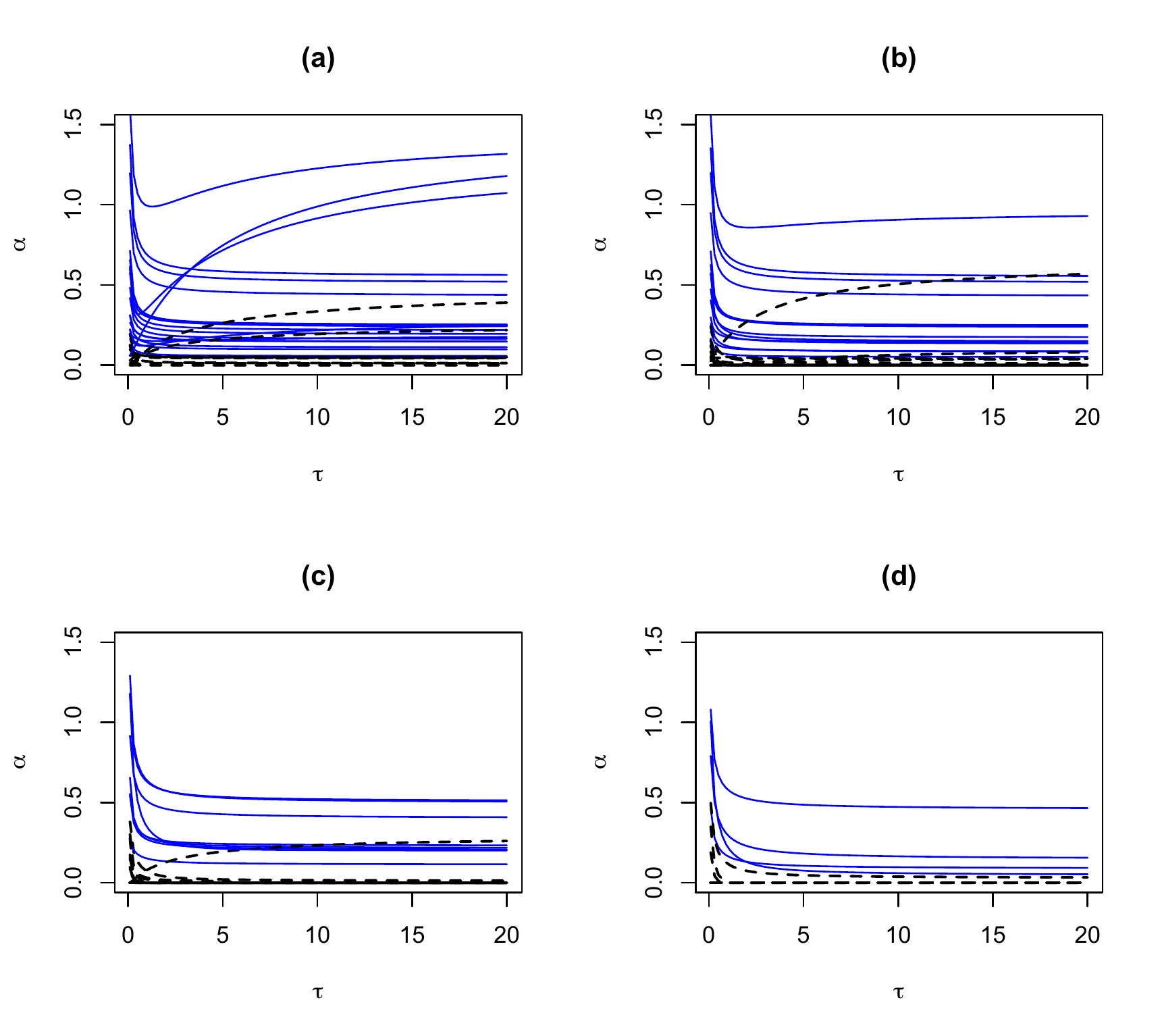}}
\caption{ { Economic growth data: The plot of the estimated parameters using the proposed method given different values of parameter $\tau$. The tuning parameter $\lambda$ is fixed at 0.0, 0.0001, 0.001, and 0.01 in plots (a), (b), (c) and (d), respectively. The blue solid (dark dashed) lines correspond to coefficients that their estimated value is (not) greater than 0.05 for the entire range of $\tau$}}
\label{fig:difftau}
\end{figure}

{ Table \ref{tab:gdata} also reports the stadard errors of the estimated effect of life expectancy using different penalization methods. The standard error of the PMOE estimator is approximated using an idea similar to \cite{chatterjee2011bootstrapping}. Specifically, we bootstrap the sample and in each bootstrap force the components of the penalized estimator $\widehat \alpha$ to zero whenever they are close to zero and estimate the treatment effect using the selected covariates, i.e., we define ${\widehat \alpha}^\dag =\hat \alpha \text{1}(|\hat \alpha|> 1/\sqrt n )$. We utilize this thresholded bootstrap method to approximate the standard error of the treatment effect. Although more investigation is required to validate the asymptotic properties of this method, we only use this standard errors to shed light on the behavior of the penalized estimators. For example, the estimators correspond to PMOE$^{\tau=5}$ and Y-fit have the lowest standard errors that may support the possibility of under-selecting important covariates.    Our results suggests that although  the effect of life expectancy is positive, it is unlikely to be significant that is consistent with \cite{acemoglu2006disease}.
}

\section{Discussion} \label{sec:conclude}

We have established a two-step procedure for estimating an unconfounded treatment effect in high-dimensional settings. First, we deal with the sparsity by penalizing a modified objective function which considers both covariate-outcome and covariate-treatment associations. Then, the selected variables are used to form a doubly robust regression estimator of the treatment effect by incorporating the propensity score in the conditional expectation of the outcome. The selected covariates may be used in other causal
techniques as well as the proposed regression method. The proposed method may also be used to identify valid instrumental variables. Specifically, one can penalize the treatment model first and record the selected variables and then apply the proposed method and identify the instrumental variables as those that are selected by the treatment model but not the proposed method.

\begin{table}[t]
\caption{The economic growth data. ATE: average treatment effect; Y-fit is obtained by penalizing the outcome model via LASSO penalty.  }\centering
\begin{tabular}{|c|ccl|} \hline
  Method   & ATE & S.D.  & C.I.($\%95$) \\ \hline
PMOE$^{\tau=0.1}$   &0.617 &0.372& (-0.127,1.361)\\
PMOE$^{\tau=0.5}$   &0.475 &0.352& (-0.229,1.179)\\
PMOE$^{\tau=5}$   &0.454 &0.340& (-0.226,1.134)\\
Cred. Reg. &0.820 &0.386&(-0.203,1.341) \\
BAC ($\omega$=$\infty$) &0.524 &0.381& (-0.228,1.286) \\
Y-fit &0.352 &0.334&(-0.111,1.345)\\ \hline
\end{tabular}
\label{tab:gdata}
\end{table}

As described in section \ref{sec:est}, any covariate selection procedure which involves the outcome variable affects the subsequent inference of the selected coefficients \citep{lee2013exact, taylor2014exact, tibshirani2014exact, taylor2015statistical, tian2015selective, lee2015model}. This is because the selected model itself is stochastic and it needs to be accounted for.   \cite{berk2012valid} proposes a method to produce a valid confidence interval for the coefficients of the selected model in the post-selection context. In our setting, although we do not penalize the treatment effect, the randomness of the selected model affects the inference about the causal effect parameter through confounding.
Moreover, note that the oracle property of the penalized regression estimators is a pointwise asymptotic feature and does not necessarily hold for all the points in the parameter space \citep{leeb2005model, leeb2008sparse}. In this manuscript, we assume that the parameter dimension ($r$) is fixed while the number of observation tends to infinity. One important extension to our work is to generalize the framework to cases where the tuple $(n,r)$ tends to infinity \citep{negahban2009unified}.  Analyzing the convergence of the estimated vector of parameters in the more general setting requires an adaptation of restricted eigenvalue condition \citep{bickel2009simultaneous} or restricted isometry property \citep{candes2007dantzig}.

\section*{Acknowledgment}
This research was supported in part by NIDA grant P50 DA010075 and NSF grant SES-1260782. The second and third authors acknowledge the support of Discovery Grants from the Natural Sciences and Engineering Research Council (NSERC) of Canada. The authors are grateful to Professor Dylan Small for enlightening discussion.

\appendixone
\section*{Appendix}

In this Appendix, we present the assumptions and proofs of the main and
other auxiliary results. We also conduct a simulation study in which either the response mean or the treatment allocation models are misspecified.

\section{Proofs of Theorem 1 \& 2}
\begin{lemma}
Suppose conditions (a) \& (b) are fulfilled, further $\lambda_n/\sqrt n \rightarrow 0$ and $\lambda_n \sqrt n \rightarrow \infty$.  Then there exists a local minimizer of the penalized objective function $M_p(.)$ for which $\| \alpha_0- \hat \alpha \|=O_p(n^{-1/2})$, where $\|.\|$ is the Euclidean norm.
\label{lem:cons}
\end{lemma}
\begin{proof}
Let $\alpha= \alpha_0+ b/\sqrt n$ for real $b$.  Define
\begin{align*}
M(b)=\frac{1}{2n} \Big[\Big| \sum_{i=1}^n\bx_i \widetilde y_i\Big|- &n(\alpha_0+b/\sqrt{n})^\top    \Big]   \Big[\Big| \sum_{i=1}^n\bx_i \widetilde y_i\Big|-     n(\alpha_0+b/\sqrt{n})^\top  \Big]^\top+  \nonumber \\
& \frac{1}{\tau}  \Big[-\Big| \sum_{i=1}^n\bx_id_i\Big|(\alpha_0+b/n) + \sum_{i=1}^n\log(1+\exp\{ \bx_i (\alpha_0+b/n) \})\Big]+ \lambda_n \sum_{j=1}^r \nu_j |\alpha_{0j}+b/\sqrt n|
\end{align*}
where $\nu_j= 1/\{( \tildealy)^2(1+| \tildeald|)^2\}$. Therefore, $\widehat V(b)=M(b)-M(0)$ is given by
\begin{align*}
\widehat V(b)=-\sum_{i=1}^n\epsilon_{i\widetilde y} b/\sqrt n+ b^\top b/2&+\frac{1}{\tau}\left[-\Big|\sum_{i=1}^n d_i \bx_i \Big|b/\sqrt n+\sum_{i=1}^n\log\left\{\frac{1+\exp\{\bx_i(\alpha_0+b/\sqrt n)  \}}{1+\exp\{\bx_i \alpha_0  \}}\right\} \right] \\
& \hspace{.5in} + \lambda_n \sum_{j=1}^r \nu_j \left[ |\alpha_{0j}+b/\sqrt n|-|\alpha_{0j}| \right]
\end{align*}
where $\sum_{i=1}^n \epsilon_{i\widetilde y}= |\sum_{i=1}^n\bx_i\widetilde y_i|- n\alpha_0^\top$. Note that the expected value of $\epsilon_{\widetilde y}$ is not necessarily zero because $\alpha_0$ corresponds to the modified objective function in (3). By applying the Taylor series expansion it can be written as
\begin{align*}
\widehat V(b)=-\sum_{i=1}^n \epsilon_{i\widetilde y}  b/\sqrt n+ b^\top b/2&+\frac{1}{\tau}\sum_{i=1}^n\left[-\epsilon_{id} b/\sqrt n+ \frac{\exp\{\bx_i\alpha_0 \}}{[1+\exp\{\bx_i\alpha_0 \}]^2} b^\top (\bx_i^\top \bx_i ) b/2 \right] \\
& \hspace{.5in} + \lambda_n \sum_{j=1}^r \nu_j \left[ |\alpha_{0j}+b/\sqrt n|-|\alpha_{0j}| \right]+o_p(n)
\end{align*}
where
\[
\sum_{i=1}^n \epsilon_{id}= |\sum_{i=1}^n d_i\bx_i|-\sum_{i=1}^n \bx_i\frac{\exp\{\bx_i\alpha_0 \}}{1+\exp\{\bx_i\alpha_0 \}} .
\]
Similar to $\epsilon_{\widetilde y}$, the expected value of $\epsilon_{d}$ is not necessarily zero. However, by (3), the expected value of $\epsilon_{\widetilde y}+\frac{1}{\tau}\epsilon_{d}$ is zero.

Using central limit theorem and laws of large numbers
\[
 \sum_{i=1}^n [\epsilon_{i\widetilde y}+\frac{1}{\tau}\epsilon_{id}] \frac{ b}{\sqrt n} \stackrel{d}{\longrightarrow} b^\top Z \qquad \qquad
 \frac{1}{n} \left[n \I+ \frac{1}{\tau} \sum_{i=1}^n \frac{\exp\{\bx_i \alpha_0\}}{[1+\exp\{\bx_i \alpha_0\}]^2} \bx_i^\top \bx_i  \right] \stackrel{p}{\longrightarrow} \Omega(\alpha_0),
\]
where $Z\sim N(0,\Sigma(\alpha_0))$ and $\I$ is the identity matrix. The behavior of the third element depends on the type of covariate. Let $\mathcal{A}$ be the set of indices of confounders and predictors of outcome. Thus $\mathcal{A}$ includes indices of variables we wish to select. If $j \in \mathcal{A}$,
$
\frac{\lambda_n}{\sqrt n} \widehat \nu_j \sqrt n \left[ |\alpha_{0j}+b_j/\sqrt n|-|\alpha_{0j}| \right] \stackrel{p}{\longrightarrow} 0.
$
Also, because $ \sum_{i=1}^n [\epsilon_{i\widetilde y}+\frac{1}{\tau}\epsilon_{id}] \frac{ b}{\sqrt n}=O_p(1)$, the order comparison of the right hand side of $\widehat V(b)$ implies that
\[
 b^\top b/2+\frac{1}{\tau}\sum_{i=1}^n\left[ \frac{\exp\{\bx_i\alpha_0 \}}{[1+\exp\{\bx_i\alpha_0 \}]^2} b^\top (\bx_i^\top \bx_i )b/2 \right]
\]
dominates the other terms. Therefore, for any given $\delta>0$, there exists a sufficiently large $\zeta_{\delta}$ such that
\begin{align*}
\lim_{n \rightarrow \infty} P\left\{  \inf_{||b||=\zeta_{\delta}} \widehat V(b) > 0 \right\} \geq 1-\delta,
\end{align*}
which implies that there is a local minimizer, say $\hat \alpha$, that satisfies $\| \alpha_0-\hat \alpha \|=O_p(n^{-1/2})$. $\blacksquare$
\end{proof}

\textbf{\it{Proof of Theorem 1}}:
\textbf{Part (a):} Lemma {\ref{lem:cons}} shows that $\widehat \alpha_j \stackrel{p}{\longrightarrow} \alpha_{0j} $, $\forall j \in \mathcal{A}$ which means $\Pr(j \in \mathcal{A}_n) \rightarrow 1$.  The proof is complete if we show that $\forall j \notin \mathcal{A}$, $\Pr(j' \notin \mathcal{A}_n) \rightarrow 1$. 

By the KKT conditions, for $\forall j' \in \mathcal{A}_n$
\[
 \frac{1}{\sqrt n} [|\sum_{i=1}^n x_{ij'} \widetilde y_i|-\sum_{i=1}^n \widehat \alpha^\top(x_{ij'} \bx_i) ]+\frac{1}{\tau\sqrt n}\left[| \sum_{i=1}^n x_{ij'}d_i|-\sum_{i=1}^n x_{ij'}\frac{\exp\{\bx_i \widehat \alpha\}}{1+\exp\{\bx_i\widehat \alpha\}} \right]=\frac{\lambda_n}{\sqrt n} \nu_{j'}. 
\]
Using central limit theorem and laws of large numbers
\[
 \sum_{i=1}^n [\epsilon_{i\widetilde y}+\frac{1}{\tau}\epsilon_{id}] \frac{ b}{\sqrt n} \stackrel{d}{\longrightarrow} b^\top Z \qquad \qquad
 \frac{1}{n} \left[n \I+ \frac{1}{\tau} \sum_{i=1}^n \frac{\exp\{\bx_i \alpha_0\}}{[1+\exp\{\bx_i \alpha_0\}]^2} \bx_i^\top \bx_i  \right] \stackrel{p}{\longrightarrow} \Omega(\alpha_0),
\]
where $Z$ has a normal distribution and $\I$ is the identity matrix. However, for a suitable choice of $\lambda_n$ (i.e., $\lambda_n/\sqrt n \rightarrow 0$ and $\lambda_n \sqrt n \rightarrow \infty$), $\frac{\lambda_n}{\sqrt n} \nu_{j'} \rightarrow_p \infty$. Thus, consistency of the estimator follows from Theorem 2 in Zou (2006).

\textbf{Part (b):} Following the results of Lemma \ref{lem:cons} and by convexity of $\widehat V(b)$ and Geyer(1994), $\arg \min \widehat V(b) \stackrel{d}{\longrightarrow} \arg \min  V(b)$ where $\arg \min \widehat V(b) = \sqrt n (\widehat \alpha-\alpha_0)$.  By Slutsky's theorem, $\arg \min  V(b_{\mathcal{A}})= \Omega_{11}^{-1}Z_{\mathcal{A}}$ where $Z_{\mathcal{A}} \sim N(0, \Sigma_{11})$ and $\arg \min  V(b_{\mathcal{A}^c})= 0$. This completes the proof of Normality. $\blacksquare$



\bigskip

\noindent \textbf{\it{Proof of Theorem 2}}:
Using the triangle inequality,
\begin{align*}
|\zeta(\widehat \theta_{M_n},M_n)-\zeta(\theta_{0},M_0)| \leq |\zeta(\widehat \theta_{M_n},M_n)-\zeta(\widehat \theta_{M_0},M_0)|+|\zeta(\widehat \theta_{M_0},M_0)-\zeta( \theta_{0},M_0)|.
\end{align*}
By differentiability of the $\zeta(.,.)$ function in $\theta$, we have $\zeta(\widehat \theta_{M_0},M_0) \CiP \zeta(\theta_{0},M_0)$, so the second term on the right hand side converges to zero in probability.  Also, $\forall t>0$, we have for each $n$ that
\begin{align*}
 \Pr(|\zeta(\widehat \theta_{M_n},M_n)-\zeta(\widehat \theta_{M_0},M_0)|>t) &= \Pr(\{M_n=M_0\} \cap \{|\zeta(\widehat \theta_{M_n},M_n)-\zeta(\widehat \theta_{M_0},M_0)|>t\}) \\
 & \quad + \Pr(\{M_n \neq M_0\} \cap \{|\zeta(\widehat \theta_{M_n},M_n)-\zeta(\widehat \theta_{M_0},M_0)|>t\})\\
 &\leq \Pr(M_n \neq M_0).
\end{align*}
We have $\Pr(M_n \neq M_0) \longrightarrow 0$ as $n \longrightarrow \infty$ by the oracle property of our procedure (Theorem 1). See also Theorem 4.2 in Wasserman and Roeder (2009).  This completes the proof of weak consistency. $\blacksquare$

\vspace{0.5in}
\textbf{\it{Proof of Proposition   1}}:
Let $\alpha^u= \alpha_0^u+b/\sqrt{n}$ and  $\nu_j= 1/\{(\tildealy)^2(1+| \tildeald|)^2\}$. Similar to the proof of Theorem 1, define
\begin{align*}
\widehat V^u(b)= -\sum_{i=1}^n \epsilon_{i\widetilde y}^u  b/\sqrt n+ b^\top b/2 &+\frac{1}{\tau}\sum_{i=1}^n \left[-\epsilon_{id}^u b/\sqrt n+ \frac{\exp\{\bu_i\alpha_0^u \}}{[1+\exp\{\bu_i\alpha_0^u \}]^2} b^\top b/2 \right] \\
& \hspace{.5in} + \lambda_n \sum_{j=1}^r \nu_j \left[ |\alpha_{0j}^u+b/\sqrt n|-|\alpha_{0j}^u| \right]+o_p(n)
\end{align*}
where $\sum_{i=1}^n \epsilon_{id}^u=\Big[|\sum_{i=1}^n d_i\bu_i|-\sum_{i=1}^n \bu_i\frac{\exp\{\bu_i\alpha_0^u \}}{1+\exp\{\bu_i\alpha_0^u \}}\Big]$ and  $\epsilon_{i\widetilde y}^u=[|\sum_{i=1}^n \bu_i \widetilde y_i|- n\alpha_0^{u\top}]$.

Let $\mathcal{A}$ be the set of indices of confounders and predictors of outcome based on the original covariate matrix $\bX$. If $j \in \mathcal{A}$,
$
\frac{\lambda_n}{\sqrt n} \nu_j \sqrt n \left[ |\alpha_{0j}^u+b_j/\sqrt n|-|\alpha_{0j}^u| \right] \stackrel{p}{\longrightarrow} 0.
$
In other words, because $\frac{\lambda_n}{\sqrt n} \rightarrow 0$, the penalty function does not impose any penalty on  $\alpha_j^u$ if $j \in \mathcal{A}$. Thus, for any finite $n$, $\widehat \alpha_j^u \neq 0$  $\forall j \in \mathcal{A}$.

Now, we show that $\Pr(\widehat \alpha_{j}^u=0; j \in \mathcal{A}^c)\rightarrow 1$ as $n\rightarrow \infty$. Let $\ell_n(\alpha_1^u,\alpha_2^u)$ be the negative log-likelihood function of $(\alpha_1^u,\alpha_2^u)$ given $\bU$.   { Also, we assume that, for each $\alpha_{0}^u$, there exists function $M_1(\bu)$ such that for $\alpha$ in the neighborhood of $\alpha_{0}^u$,
\[
\left| \frac{\partial f(\bu,\alpha)}{\partial \alpha} \right| \leq M_1(\bu),
\]
such that $\int M_1(\bu) d\bu < \infty$.  By the mean value theorem,
\[
\ell_n(\alpha_1^u,\alpha_2^u)-\ell_n(\alpha_1^u,0)=\left[ \frac{\partial \ell_n(\alpha_1^u,\xi)}{\partial \alpha_2^u}  \right]^\top  \alpha_2^u,
\]
for some $ \| \xi \| \leq \| \alpha_2^u \|$. Also, by mean value theorem and adding and subtracting $\frac{\partial \ell_n(\alpha_1,0)}{\partial \alpha_2^u}$, we have 
\[
\left\|\frac{\partial \ell_n(\alpha_1^u,\xi)}{\partial \alpha_2^u}- \frac{\partial \ell_n(\alpha_{01}^u,0)}{\partial \alpha_2^u} \right\|^2 \leq \left[ \sum_{i=1}^n M_1(x_i)  \right] \|\xi\| +\left[ \sum_{i=1}^n M_1(x_i)  \right] \|\alpha_{01^u}-\alpha_1^u   \| .
\]
Let $\mathcal{B}$ be a set of indices for which $\alpha_j^u \neq 0$ while $\alpha_j = 0$.  For $j \in \mathcal{B}$, $ \| \xi \| \leq \| \alpha_j^u \|=O_p(1)$. Thus
\[
\left\|\frac{\partial \ell_n(\alpha_1^u,\xi)}{\partial \alpha_2^u}- \frac{\partial \ell_n(\alpha_{01}^u,0)}{\partial \alpha_2^u} \right\|^2 \leq O_p(n).
\]
Because, for $j \in \mathcal{B}$, $\frac{\partial \ell_n(\alpha_{01}^u,0)}{\partial \alpha_2^u}= O_p(n) $,  we conclude that $\frac{\partial \ell_n(\alpha_1^u,\xi)}{\partial \alpha_2^u}  = O_p(n)$.

Let $\ell_n^p(\alpha_1^u,\alpha_2^u)=\ell_n(\alpha_1^u,\alpha_2^u) + \lambda_n \sum_{j \in \mathcal{A} \cup \mathcal{A}^c } \widehat \nu_j |\alpha_j^u|$.  By applying the assessed orders, we have
\begin{align}
\ell_n^p(\alpha_1,\alpha_2)-\ell_n^p(\alpha_1,0) &= \sum_{j \in  \mathcal{A}^c } \{ -|\alpha_{j}|O_p(n)+ \lambda_n \widehat \nu_j |\alpha_j| \} \nonumber\\ &=\sum_{j \in  \mathcal{A}^c } \{ -|\alpha_{j}|O_p(n)+ n \lambda_n O_p(1) |\alpha_j| \}.
\label{eq:ineq}
\end{align}
Thus,
\begin{align*}
\ell_n^p(\alpha_{\mathcal{A}},\alpha_{\mathcal{A}^c})-\ell_n^p(\widehat \alpha_{\mathcal{A}},0) \pm \ell_n^p( \alpha_{\mathcal{A}},0)   \geq \ell_n^p(\alpha_{\mathcal{A}},\alpha_{\mathcal{A}^c})-\ell_n^p(\alpha_{\mathcal{A}},0).
\end{align*}
By (\ref{eq:ineq}), the RHS of the above inequality is positive with probability 1 when $\lambda_n  \rightarrow  \infty$ as $n \rightarrow \infty$. $\blacksquare$

\section{Non-orthogonal covariates}
\label{app:sim2}

In this section, we study the performance of our  covariate selection and estimation procedure when the covariates are non-orthogonal.
We generate 500 data sets of sizes 300 and 500 from the following two models:
\begin{enumerate}
\item[] $D \sim \text{Bernoulli}\left(\dfrac{\exp\{x_1-2x_2+x_5-x_6+x_7-x_8}
    {1+\exp\{x_1-2x_2+x_5-x_6+x_7-x_8\}}\right)$, \\[12pt]
$Y \sim \text{Normal}(d+x_1+0.2x_2-x_3+x_4, 2)$
\end{enumerate}
\noindent where $\bX_k \sim N(1,\sqrt2)$ for $k=1,...,100$. The covariate $x_2$ is considered as a weak confounder.

Tables 5 \& 6 show that the proposed method is still outperforming the other methods. Specifically, the variance of the estimator in the Bayesian credible region (Cred. Reg.) is large due to the inclusion of spurious variables that are not related to the outcome (Table 6). This is the cause of bias in the Cred. Reg. estimator for small sample size $n=300$.

\begin{table}[t]
\caption{ Simulation results for non-orthogonal covariates. Bias, S.D. and MSE are for the treatment effect. S.D: empirical standard error. Y-fit is obtained by penalizing the outcome model via LASSO penalty.}
\label{tab:largep}
\begin{center}
\begin{tabular}{lrrr|rrrrr|} \hline
& \multicolumn{3}{c}{$n=300$} & \multicolumn{3}{c}{$n=500$} \\
Method   &\multicolumn{1}{c}{Bias} & \multicolumn{1}{c}{S.D}  & \multicolumn{1}{c|}{MSE} &
\multicolumn{1}{c}{Bias} & \multicolumn{1}{c}{S.D} & \multicolumn{1}{c}{MSE}  \\ \hline
PMOE$^{\tau=0.5}$      &    0.013 &0.365 &0.133&0.028 &0.265 &0.071         \\
Cred. Reg.  &0.169 &0.374 & 0.168&0.050 &0.305  &0.096  \\
Y-fit   &         0.265 &0.375  &0.211&0.190&0.359  &0.169 \\
Oracle   &         0.031 &0.300 & 0.091&0.015& 0.222  &0.050 \\ \hline
\end{tabular}
\end{center}
\end{table}

\begin{table}[t]
\caption{ { { Simulation results for non-orthogonal covariates. Number of coefficients that are correctly or incorrectly set to zero. Y-fit is obtained by penalizing the outcome model via LASSO penalty.  }}}\centering
\begin{tabular}{c c c  |c c c c c c} \hline
Method  & \multicolumn{1}{c}{Correct} & \multicolumn{1}{c}{Incorrect}
 & \multicolumn{1}{c}{Correct} & \multicolumn{1}{c}{Incorrect}\\
\hline
  & \multicolumn{2}{c}{$n=300$} & \multicolumn{2}{c}{$n=500$}\\
PMOE$^{\tau=0.5}$   &94.24 &0.12   &  95.08 &0.08   \\
Cred. Reg.   &90.03 &0.00  &   90.04 &0.00    \\
 Y-fit      &97 &0.71   &   97 &0.55   \\           \hline
\end{tabular}
\label{tab:largep2}
\end{table}


\section{Performance under model misspecification}
\label{app:sim}
In this simulation study, we want to examine the performance of our proposed method when 1) either of the outcome or treatment working models are misspecified and 2) the number of potential confounders ($r$) is larger then the sample size.

\begin{enumerate}
\item[1.] $D \sim \text{Bernoulli}\left(\dfrac{\exp\{0.5x_1-x_2+0.5x_5-0.5x_6+0.5x_7\}}{1+\exp\{0.5x_1-x_2+0.5x_5-0.5x_6+0.5x_7\}}\right)$ \\[12pt]
$Y \sim \text{Normal}(d+x_1+0.2x_2+3x_3+3x_4, 2)$

\medskip

\item[2.] $D \sim \text{Bernoulli}\left(\dfrac{\exp\{0.1x_1+x_2+0.7 \frac{x_{10}+x_{9}}{1+|x_8|}\}}{1+\exp\{0.1x_1+x_2+0.7 \frac{x_{10}+x_{9}}{1+|x_8|}\}}\right)$ \\[12pt]
$Y \sim \text{Normal}(d+x_1+0.2x_2+3x_3+3x_4, 2)$

\medskip

\item[3.] $D \sim \text{Bernoulli}\left(\dfrac{\exp\{0.5x_1-x_2+0.5x_5-0.5x_6+0.5x_7\}}{1+\exp\{0.5x_1-x_2+0.5x_5-0.5x_6+0.5x_7\}}\right)$ \\[12pt]
$Y \sim \text{Normal}\left(d+x_1+2 \dfrac{\exp\{0.2x_3+0.2x_4\}}{\exp\{0.2|x_1|+0.2|x_2|\}}, 2 \right)$
\end{enumerate}
where $\bX_k \sim N(0,2)$ for $k=1,...,550$.

In all the scenarios, we consider linear working models for the treatment and outcome. Thus, in scenarios 2 \& 3, at least one of them is misspecified.  Table \ref{tab:miss} summarized the results. Y-fit refers to the estimator obtained by penalizing the outcome model using SCAD penalty. 

In scenarios 1 \& 2, $x_2$ is a non-ignorable confounder which is weakly associated with the outcome. Ignoring this variable by Y-fit method results in bias which does not go zero by increasing the sample size.  Our proposed method PMOE outperforms Y-fit by increasing the chance of including all the confounders (weak or strong) in the model. Our simulation also shows that the proposed method selects important covariates even when the true generative models  are non-linear. It also highlights the importance of  using a double robust estimator in the estimation step  to obtain a consistent treatment effect estimate when at least one of the outcome or treatment models are correctly specified.

\begin{table}
 \caption{ Performance of the proposed method when either of the outcome or treatment models are misspecified and $r>n$.}
 \label{tab:miss}
\begin{center}
\begin{tabular}{lrrr|rrr|} \hline
Method   &\multicolumn{1}{c}{Bias} & \multicolumn{1}{c}{S.D.} & \multicolumn{1}{c|}{MSE} &
\multicolumn{1}{c}{Bias} & \multicolumn{1}{c}{S.D.} & \multicolumn{1}{c}{MSE}  \\ \hline
Scenario 1.               & \multicolumn{3}{c}{$n=300$} & \multicolumn{3}{|c|}{$n=500$} \\
PMOE$^{\tau=0.5}$    &  0.010  &0.304 &0.092&0.005 &0.222&  0.049          \\
Y-fit   &         0.301 &0.343 &0.208&0.243& 0.325 &0.165 \\
Oracle   &         0.006 &0.285 &0.081&0.003& 0.203 &0.041 \\ \hline
Scenario 2.               & \multicolumn{3}{c}{$n=300$} & \multicolumn{3}{|c|}{$n=500$} \\
PMOE$^{\tau=0.5}$    &  0.026  &0.336 &0.114&0.004 &0.260&  0.068          \\
Y-fit   &         0.060 &0.407 &0.169&0.087& 0.348 &0.129 \\
Oracle   &         0.020 &0.322 &0.104&0.003& 0.258 &0.067 \\ \hline
Scenario 3.               & \multicolumn{3}{c}{$n=300$} & \multicolumn{3}{|c|}{$n=500$} \\
PMOE$^{\tau=0.5}$    &  0.051  &0.363 &0.134&0.012 &0.264&  0.070          \\
Y-fit    &         0.091 &0.373 &0.147&0.112 &0.311 &0.109 \\
Oracle   &         0.002 &0.313 &0.098&0.009& 0.225 &0.051 \\ \hline
\end{tabular}
\end{center}
\end{table}

\bigskip

\bibliographystyle{Biometrika}
\bibliography{mybib-1}

\end{document}